\documentclass[12pt]{article}
\usepackage[a4paper, total={6.5in, 9.5in}]{geometry}
\usepackage{amsthm,amsmath,amssymb}
\usepackage{algorithm}
\usepackage{algpseudocode}
\usepackage{cleveref}
\usepackage{authblk}
\newcommand{\dpf}{{\tt path}}
\newcommand{\minham}{{d_{\min}}}
\newcommand{\sumham}{{d_{\rm sum}}}

\newtheorem{theorem}{Theorem}
\newtheorem{lemma}{Lemma}
\newtheorem{corollary}{Corollary}
\newtheorem{observation}{Observation}

\title{Finding Diverse Trees, Paths, and More}
\author[1]{Tesshu Hanaka}
\author[2]{Yasuaki Kobayashi}
\author[3]{Kazuhiro Kurita}
\author[4]{Yota Otachi}

\affil[1]{Chuo University}
\affil[2]{Kyoto University}
\affil[3]{National Institute of Informatics}
\affil[4]{Nagoya University}
\date{}

\begin{document}

\maketitle

\begin{abstract}
Mathematical modeling is a standard approach to solve many real-world problems and {\em diversity} of solutions is an important issue, emerging in applying solutions obtained from mathematical models to real-world problems.
Many studies have been devoted to finding diverse solutions.
Baste et al. (Algorithms 2019, IJCAI 2020) recently initiated the study of computing diverse solutions of combinatorial problems from the perspective of fixed-parameter tractability.
They considered problems of finding $r$ solutions that maximize some diversity measures (the minimum or sum of the pairwise Hamming distances among them) and gave some fixed-parameter tractable algorithms for the diverse version of several well-known problems, such as {\sc Vertex Cover}, {\sc Feedback Vertex Set}, {\sc $d$-Hitting Set}, and problems on bounded-treewidth graphs.
In this work, we further investigate the (fixed-parameter) tractability of problems of finding diverse spanning trees, paths, and several subgraphs.
In particular, we show that, given a graph $G$ and an integer $r$, the problem of computing $r$ spanning trees of $G$ maximizing the sum of the pairwise Hamming distances among them can be solved in polynomial time. 
To the best of the authors' knowledge, this is the first polynomial-time solvable case for finding diverse solutions of unbounded size.
\end{abstract}

\section{Introduction}

To solve real-world problems, we often formulate problems as mathematical models and then develop algorithms working on these mathematical models.
In this context, algorithms are usually designed to find a {\em single} (near) optimal solution by optimizing an objective function formulated in a mathematical model.
However, such a solution may be inadequate for original real-world problems since mathematical models ``approximately'' formulate them and some tacit rules and ambiguous constraints inherent in real-world problems are usually ignored for taming mathematical models.
One possible (and straightforward) solution to this issue is to find {\em multiple} solutions rather than a single solution.
The concept of {\em $k$-best enumeration}~\cite{Eppstein:k-best:2016} is a promising approach along this line.
In this approach, given a parameter $k$, we compute a set of $k$ distinct solutions, such that every solution in the set is better than any solutions not in the set.
This is a natural extension of usual optimization problems, and many algorithms for finding $k$-best solutions have been developed for classical combinatorial optimization problems~\cite{Murty:Letter:1968,Lawler:procedure:1972,Gabow:Two:1977} and other problems arise in several fields~\cite{Hara:Enumerate:2017,Lindgren:Exact:2017}.

However, solutions obtained by $k$-best enumeration algorithms might be similar to each other since those solutions are essentially made by some local modifications.
The basic strategy of $k$-best enumeration algorithms is that we first find a single optimal solution $S = \{x_1, \ldots, x_t\}$ and then, for each $1 \le i \le t$, compute an optimal solution that includes $\{x_1, \ldots, x_{i-1}\}$ but excludes $x_i$.
The entire algorithm recursively computes $k$ solutions in this way.
Since the subsequent solution must contain $\{x_1, \ldots x_{i-1}\}$ and also tends to contain $\{x_{i+1}, \ldots, x_{t}\}$, these solutions would be similar, which would not fit for our original purpose.

To address this issue, {\em diversity} among solutions is an important factor, and a substantial effort is dedicated to finding multiple solutions that optimize diversity measures in the literature \cite{Hebrard:Finding:2005,Danna:How:2009,Petit:Finding:2015,Baste:FPT:2019,Baste:Diversity:2020,Ingmar:Modelling:2020}.
Let $U$ be a finite set and let $\mathcal S \subseteq 2^U$ be the set of solutions in a combinatorial optimization problem.
Several diversity measures have been proposed.
Among others, many existing studies focus on maximizing the pairwise Hamming distances among solutions.
In particular, the following two diversity measures are widely used.
\begin{eqnarray*}
    \displaystyle\sumham(U_1, \ldots, U_r) &=& \sum_{1 \le i < j \le r} (|U_i \setminus U_j| + |U_j \setminus U_i|),\\
    \displaystyle\minham(U_1, \ldots, U_r) &=& \min_{1 \le i < j \le r} (|U_i \setminus U_j| + |U_j \setminus U_i|),
\end{eqnarray*}
where $U_1, \ldots, U_r \in \mathcal S$ are not necessarily distinct.


Baste et al. \cite{Baste:FPT:2019,Baste:Diversity:2020} initiated the study of finding diverse solutions from the perspective of fixed-parameter tractability.
{\em Fixed-parameter tractability} is the central notion in parameterized complexity theory, which extends the notion of tractability (i.e., polynomial-time solvability) in classical complexity theory.
In this context, we consider problems that take an instance $I$ and a {\em parameter} $k$, and analyze the complexity of those problems in terms of the input size $|I|$ and parameter $k$.
We say that a problem is {\em fixed-parameter tractable (FPT) parameterized by $k$} if it admits an algorithm that runs in time $f(k)|I|^{O(1)}$, where $f$ is a computable function. See~\cite{Downey:Fundamentals:2013,Cygan:Parameterized:2015} for more information.

Baste et al.~\cite{Baste:FPT:2019,Baste:Diversity:2020} studied the diverse version of several problems that have been intensively studied in the parameterized complexity community.
They show that {\sc Diverse Vertex Cover} and {\sc Min-Diverse Vertex Cover} are fixed-parameter tractable parameterized by $k + r$, where {\sc Diverse Vertex Cover} and {\sc Min-Diverse Vertex Cover} asks for $r$ vertex covers $C_1, \ldots, C_r$ of $G$ with size at most $k$ that maximize $\sumham(C_1, \ldots, C_r)$ and $\minham(C_1, \ldots C_r)$, respectively.
More specifically, they gave $2^{kr}r^2n^{O(1)}$-time algorithms for both problems~\cite{Baste:FPT:2019}.
They also gave fixed-parameter algorithms for {\sc (Min-)Diverse Feedback Vertex Set} and {\sc (Min)-Diverse $d$-Hitting Set}~\cite{Baste:FPT:2019} parameterized by $k + r$ and $k + r + d$, respectively (See \cite{Baste:FPT:2019} for details).
Baste et al.~\cite{Baste:Diversity:2020} gave a general framework to maximize $\sumham$ for bounded-treewidth graphs.
Roughly speaking, on $n$-vertex graphs of treewidth $t$, they showed that if a single solution can be found in time $f(t)n^{O(1)}$, then one can find $r$ solutions $U_1, \ldots, U_r$  maximizing $\sumham(U_1, \ldots, U_r)$ in time $f(t)^rn^{O(1)}$. 

\paragraph{Our Contributions}
In this paper, we further investigate the fixed-parameter tractability of problems of finding diverse solutions for several combinatorial problems. In particular, we consider the following problems:
\begin{itemize}
    \item {\sc Diverse Spanning Tree}: Given a graph $G$ and an integer $r$, the problem asks for $r$ spanning trees $T_1, \ldots, T_r$ in $G$ maximizing $\sumham(E(T_1), \ldots, E(T_r))$.
    \item {\sc Diverse $k$-Path} and {\sc Min-Diverse $k$-Path}: Given a graph $G$ and integers $k$ and $r$, the problem asks for $r$ paths $P_1, \ldots, P_r$ of length $k - 1$ in $G$ maximizing $\sumham(E(P_1), \ldots, E(P_r))$ and $\minham(E(P_1), \ldots, E(P_r))$, respectively.
    \item {\sc Diverse Matching} and {\sc Min-Diverse Matching}: Given a graph $G$ and integers $k$ and $r$, the problems ask for $r$ matching $M_1, \ldots, M_r$ of size $k$ in $G$ maximizing $\sumham(M_1, \ldots, M_r)$ and $\minham(M_1, \ldots, M_r)$, respectively.
    \item {\sc Diverse Subgraph Isomorphism} and {\sc Min-Diverse Subgraph Isomorphism}: Given graphs $G, H$ and an integer $r$, the problems ask for $r$ subgraphs $H_1, \ldots, H_r$ isomorphic to $H$ in $G$ maximizing $\sumham(E(H_1), \ldots, E(H_r))$ and $\minham(E(H_1), \ldots, E(H_r))$, respectively.
\end{itemize}

We show that {\sc Diverse Spanning Tree} is polynomial-time solvable.
The result can be extended to arbitrary matroids: we give a polynomial-time algorithm that, given a matroid $\mathcal M$ with an independence oracle and an integer $r$, computes $r$ bases $B_1, \ldots, B_r$ of $\mathcal M$ maximizing $\sumham(B_1, \ldots, B_r)$.
It is worth mentioning that several studies (e.g.~\cite{Abbassi:Diversity:2013,Borodin:Max-Sum:2017}) consider diversity maximization problems under matroid constraints and gave approximation algorithms for them.
However, our problems, {\sc Diverse Spanning Tree} and its generalization, are essentially different from those problems.
We also give a general framework to obtain diverse $r$ solutions of size $k$.
An illustrative example of this result is {\sc Diverse $k$-Path}.
The problem of finding a $k$-path (that is, a path of length $k - 1$) is one of the best-studied problems in parameterized algorithms.
The seminal work of Alon et al.~\cite{Alon:Color:1995} showed this problem can be solved in time $(2e)^kn^{O(1)}$, where $e$ is the base of the natural logarithm.
Their technique, {\em color-coding}, is of great importance for finding ``patterns'' of size $k$.
We exploit this technique to find diverse solutions.
Our general framework states that if we can find a single ``colorful pattern'' in time $f(k)n^{O(1)}$, then we can find diverse solutions in time $f(k, r)n^{O(1)}$ as well.
As applications, we show several fixed-parameter tractable algorithms for diverse problems, including {\sc (Min-)Diverse $k$-Path}, {\sc (Min-)Diverse Matching}, {\sc (Min-)Diverse Subgraph Isomorphism} of bounded-treewidth pattern graphs.

Very recently, Fomin el al.~\cite{Fomin:Diverse:2020} studied the problem of finding a pair of maximum matchings $M_1$ and $M_2$ with $\sumham(M_1, M_2) \ge k$ and showed that this problem is solvable in polynomial time for an arbitrary $k$ on bipartite graphs and fixed-parameter tractable parameterized by $k$ on general graphs.
Although their problem is restricted to the case $r = 2$, the size of matchings is not considered as a parameter.


\section{Preliminaries}
For an integer $k \ge 1$, we use $[k]$ to denote $\{1, 2, \ldots, k\}$.

Let $G$ be an undirected graph.
We denote by $V(G)$ and by $E(G)$ the sets of vertices and edges of $G$, respectively.
For a vertex $v \in V(G)$, the set of neighbors of $v$ is denoted by $N_G(v)$ (i.e., $N_G(v) = \{w \in V(G): \{v, w\} \in E(G)\}$).
For $F \subseteq E(G)$, we denote by $G[F]$ the subgraph of $G$ consisting of edges in $F$: $G[F] = (V(F), F)$, where $V(F)$ is the set of end vertices of $F$.

Let $U$ be a finite set and let $k$ be a positive integer.
A function $c: U \to [k]$ is called a {\em coloring} of $U$ and each integer in $[k]$ is called a {\em color}.
For a set of colors $C \subseteq [k]$, a subset of $U$ is said to be {\em $C$-colorful} ({\em with respect to $c$}) if every element in the subset has a distinct color and the set of colors used in the subset is exactly $C$.
We frequently consider vertex-colorings on graphs.
Let $H$ be a subgraph of $G = (V, E)$.
Consider a coloring $c: V \to [k]$ on vertices.
Note that $c$ is not necessarily proper, which means there may be two adjacent vertices with the same color.
For $C \subseteq [k]$, we say that $H$ is $C$-colorful if $V(H)$ is $C$-colorful.
We also say that $H$ is {\em colorful} if it is $C$-colorful for some $C \subseteq [k]$.

Let $E$ be a finite set and let $\mathcal I \subseteq 2^{E}$ be a collection of subsets of $E$.
A pair $(E, \mathcal I)$ is called a {\em matroid} if the following three axioms hold: (M1) $\emptyset \in \mathcal I$, (M2) for $X, Y \subseteq E$ with $X \subseteq Y$, $Y \in \mathcal I$ implies $X \in \mathcal I$, and (M3) for $X, Y \in \mathcal I$ with $|X| < |Y|$, there is $e \in Y \setminus X$ such that $X \cup \{e\} \in \mathcal I$.
Each set in $\mathcal I$ is called an {\em independent set} and each inclusion-wise maximal set in $\mathcal I$ is called a {\em base} of $\mathcal I$.
It is easy to see that all bases of a matroid have the same cardinality. 
The {\em rank} of a matroid is the cardinality of a base of the matroid.
Note that $\mathcal I$ may have exponentially many subsets of $E$ in general.
However, algorithms described in this paper still run in time $|E|^{O(1)}$ when we are given $\mathcal I$ as an independence oracle that is assumed to be evaluated in time $|E|^{O(1)}$.
Given this, we say that an algorithm for problems on matroids runs in polynomial time if it runs in time $|E|^{O(1)}$ under this assumption.

Let $G = (V, E)$ be a multigraph (i.e., $G$ may contain parallel edges).
Define $\mathcal I_G$ as the collection of all subsets $F \subseteq E$ such that $G[F]$ has no cycles.
Then, pair $(E, \mathcal I_G)$ satisfies the above three conditions and hence it is a matroid, called a {\em graphic matroid}.
It is easy to see that if $G$ is connected, every base of $(E, \mathcal I_G)$ is a spanning tree of $G$.

\section{Finding Diverse Spanning Trees}
Let $G = (V, E)$ be a connected graph and let $r$ be an integer.
The goal of this section is to develop an algorithm for finding (not necessarily edge-disjoint) $r$ spanning trees $T_1, T_2, \ldots, T_r$ maximizing $\sumham(E(T_1), \ldots, E(T_r))$. 

\begin{theorem}\label{thm:diverse-tree}
    There is a polynomial-time algorithm that, given a graph $G$ and an integer $r$, computes $r$ spanning trees $T_1, \ldots, T_r$ of $G$ maximizing $\sumham(E(T_1, \ldots, E(T_r)))$.
\end{theorem}

The problem can be translated into words in matroid theory: Find $r$ bases $B_1, \ldots, B_r$ of a graphic matroid $(E, \mathcal I)$ that maximizes $\sumham(B_1, \ldots, B_r)$.
As a special case of our problem, the problem of finding {\em mutually disjoint} $r$ bases of a matroid can be solved in polynomial time by a greedy algorithm for the matroid union, provided that the membership of $\mathcal I$ can be decided in polynomial time.
More generally, the weighted version of this problem can be solved in polynomial time.

\begin{theorem}[\cite{Nash-Williams:application:1967}]\label{thm:tree:packing}
    Let $\mathcal M = (E, \mathcal I)$ be a matroid and let $w: E \to \mathbb R$.
    Suppose that the membership of $\mathcal I$ can be checked in polynomial time.
    Then, the problem of deciding whether there is a set of mutually disjoint $r$ bases $B_1, \ldots, B_r$ of $\mathcal M$ can be solved in polynomial time.
    Moreover, if the answer is affirmative, we can find such bases that minimize the total weight (i.e., $\sum_{1 \le i \le r} \sum_{e \in B_r} w(e)$) in polynomial time.
\end{theorem}

Our objective is to maximize $\sumham(E(T_1), \ldots, E(T_r))$, where $T_i$ is a spanning tree of $G$ for $1 \le i \le r$.
This can be rewritten as:
\[
    \sum_{1 \le i < j \le r} (|E(T_i)| + |E(T_j)| - 2 |E(T_i) \cap E(T_j)|).
\]
As $|E(T_1)| = \cdots = |E(T_r)| = |V| - 1$, the problem is equivalent to that of minimizing the pairwise sum of $|E(T_i) \cap E(T_j)|$, that is, finding a (not necessity edge-disjoint) $r$ spanning trees $T_1, \ldots, T_r$ minimizing
\begin{equation}\label{eq:tree:obj}
    \sum_{1 \le i < j \le r} |E(T_i) \cap E(T_j)|.
\end{equation}

To solve this minimization problem, we reduce it to the problem of finding $r$ disjoint bases of a graphic matroid with minimum total weight, which can be solved in polynomial time (\Cref{thm:tree:packing}).
For each edge $e$ in $G$, we replace it with $r$ parallel edges $e_1, \ldots, e_r$, and the obtained multigraph is denoted by $G' = (V, E')$.
The weight of edges in $G'$ is defined as follows.
For each $e \in E$ and $1 \le i \le r$, the weight of $e_i$ is defined as $w(e_i) = i - 1$.
Clearly, the construction of $G'$ can be done in polynomial time.

\begin{lemma}\label{lem:tree:reduction}
    Let $k$ be an integer.
    Then, $G$ has $r$ spanning trees $T_1, \ldots, T_r$ with $\sum_{1 \le i < j \le r} |E(T_i) \cap E(T_j)| \le k$ if and only if there is a set of disjoint $r$ spanning trees in $G'$ whose total weight is at most $k$.
\end{lemma}

\begin{proof}
    Suppose that $G$ has $r$ spanning trees $T_1, \ldots, T_r$ with $\sum_{1 \le i < j \le r} |E(T_i) \cap E(T_j)| \le k$.
    For each $e \in E$, we let $m(e) = |\{i \in [r] : e \in E(T_i)\}|$.
    Let $F = \{e_1, e_2, \ldots, e_{m(e)} : e \in E\}$.
    Observe that $F$ can be partitioned into $r$ disjoint spanning trees $T'_1, \ldots, T'_r$ of $G'$: For each $1 \le i \le r$ and $e \in E(T_i)$, $T'_i$ contains one of $e_1, \ldots, e_{m(e)}$ that is not contained in any other $T'_j$.
    This can be done by the definition of $m(e)$.
    Moreover, the total weight of $F$ is
    \begin{equation*}
        \sum_{e \in E} \sum_{1 \le i \le m(e)} w(e_i)  = \sum_{e \in E} \sum_{1 \le i \le m(e)} (i-1) = \sum_{e \in E} \binom{m(e)}{2}.
    \end{equation*}
    Since $e$ is contained in $m(e)$ trees of $T_1, \ldots, T_r$, the contribution of $e$ to $\sum_{1\le i < j \le r}|E(T_i) \cap E(T_j)|$ is exactly $\binom{m(e)}{2}$.
    Therefore, the total weight of $r$ disjoint spanning trees is at most $k$.
    
    Conversely, let $T'_1, \ldots, T'_r$ be $r$ disjoint spanning trees of $G'$ minimizing the total weight.
    Note that as $G$ is connected and for each pair of adjacent vertices in $G'$, there are $r$ parallel edges between them, we can always find these disjoint spanning trees from $G'$. 
     Let $F = \bigcup_{1 \le i \le r} E(T'_i)$ and let $k = \sum_{e \in F} w(e)$.
    For each $e \in E$, observe that either $F$ does not contains any of $e_1, e_2, \ldots, e_r$ or contains $e_1, e_2, \ldots, e_{j}$ for some $1 \le j \le r$.
    This follows from the fact that if $e_j \notin F$ and $e_{j'} \in F$ for some $j < j'$, we can exchange $e_{j'}$ with $e_j$, strictly decreasing the total weight of $F$.
    For each $e \in E$, we let $m(e) = |\{e_1, \ldots, e_r\} \cap F|$.
    Then, the total weight of $F$ is $\sum_{e \in E} \binom{m(e)}{2}$.
    Let $T_1, \ldots T_r$ be the spanning trees of $G$ such that $e \in E$ is contained in $T_i$ if and only if some edge $e_j$ copied from $e$ is contained in $T'_i$. 
    Since $e$ is contained in $m(e)$ spanning trees among $T_1, \ldots, T_r$, it contributes $\binom{m(e)}{2}$ to the objective function (\ref{eq:tree:obj}).
    Therefore, we have $\sum_{1 \le i < j \le r} |E(T_i) \cap E(T_j)| = k$.
\end{proof}

Now, we are ready to describe our algorithm.
Given a graph $G$ and an integer $r$, we compute the multigraph $G'$ and the edge weight function $w$ defined as above.
Let $E'$ be the set of edges of $G'$ and $\mathcal I = \{F \subseteq E' : G'[F] \text{ has no cycles}\}$.
Then, $(E', \mathcal I)$ is a graphic matroid.
By \Cref{thm:tree:packing}, we can compute $r$ disjoint spanning trees in $G'$, whose total weight is minimized, in polynomial time.
By \Cref{lem:tree:reduction}, we can find $r$ spanning trees $T_1, \ldots, T_r$ in $G'$ maximizing $\sumham(E(T_1), \ldots, E(T_r))$ in polynomial time, which completes the proof of \Cref{thm:diverse-tree}.

The above argument can be extended to an arbitrary matroid equipped with an independent oracle that can be evaluated in polynomial time.
To see this, we need the following lemma.
\begin{lemma}\label{lem:matroid:extension}
    Let $(E, \mathcal I)$ be a matroid, $e \in E$, and $e' \notin E$.
    Define $\mathcal I' = \mathcal I \cup \{(F \setminus \{e\}) \cup \{e'\} : F \in \mathcal I, e \in F\}$.
    Then, the pair $(E \cup \{e'\}, \mathcal I')$ is a matroid.
    Moreover, if the membership of $\mathcal I$ can be decided in polynomial time, then so is $\mathcal I'$.
\end{lemma}
\begin{proof}
    It is well known that $(E \cup \{e'\}, \mathcal I')$ is a matroid~\cite{Nagamochi:Complexity:1997}.
    
    To check the membership of $\mathcal I'$ for given $F \subseteq E \cup \{e'\}$, we test whether $\{e, e'\} \subseteq F$, $F \in \mathcal I$ for the case $e' \notin F$, and $(F \setminus \{e'\}) \cup \{e\} \in \mathcal I$ for the case $e' \in F$.
    These can be done in polynomial time.
\end{proof}

To find $r$ bases $B_1, \ldots, B_r$ of a matroid $\mathcal M = (E, \mathcal I)$ maximizing $\sumham(B_1, \ldots, B_r)$,
we perform a polynomial-time reduction to the problem of finding minimum weight $r$ disjoint bases of a matroid and solve it by the algorithm in \Cref{thm:tree:packing} as in \Cref{lem:tree:reduction}.
For each element $e \in E$, we let $e_1, \ldots, e_r$ be $r$ copies of $e$, and let $E' = \{e_i : e \in E, 1 \le i \le r\}$.
Since the size of bases of $\mathcal M$ is its rank, denoted by $r(\mathcal M)$, and
\[
    \sumham(B_1, \ldots, B_r) = \sum_{1 \le i < j \le r} 2(r(\mathcal M) - |B_i \cap B_j|),
\]
the objective is to minimize $\sum_{1 \le i < j \le r} |B_i \cap B_j|$.
Define $\mathcal I' \subseteq 2^{E'}$ as follows.
For $F \subseteq E$, we let $\mathcal C_F \subseteq 2^{E'}$ be the set of all $F' \subseteq E'$ such that $F'$ contains exactly one of $e_1, \ldots, e_r$ for each $e \in F$.
Then, $\mathcal I' = \bigcup_{F \in \mathcal I}\mathcal C_F$.
By~\Cref{lem:matroid:extension}, $\mathcal M' = (E', \mathcal I')$ is a matroid.
By weighting each element of $E'$ as $w(e_i) = i - 1$ for each $e \in F$ and $1 \le i \le r$, similarly to \Cref{lem:tree:reduction}, we can prove that $\mathcal M$ has $r$ bases $B_1, \ldots, B_r$ such that $\sum_{1 \le i < j \le r} |B_i \cap B_j|\le k$ if and only if $\mathcal M'$ had $r$ disjoint bases of total weight at most $k$, which can be found in polynomial time.

\begin{theorem}\label{thm:matroid:ptime}
    Let $\mathcal M = (E, \mathcal I)$ be a matroid. Suppose that the membership of $\mathcal I$ can be checked in polynomial time.
    Then, we can find $r$ bases $B_1, \ldots, B_r$ of $\mathcal M$ maximizing $\sumham(B_1, \ldots, B_r)$ in polynomial time.
\end{theorem}

Let $\mathcal M = (E, \mathcal I)$ be a matroid and let $k$ be a positive integer.
The {\em $k$-truncation} of $\mathcal M$ is the pair $(E, \mathcal I')$ such that $F \subseteq E$ belongs to $\mathcal I'$ if and only if $|F| \le k$ and $F \in \mathcal I$.
It is known that the $k$-truncation of a matroid is also a matroid~\cite{Oxley:Matroid:2006}.
Hence, by \Cref{thm:matroid:ptime}, we have the following corollary.

\begin{corollary}
    Let $G = (V, E)$ be a graph and $k, r$ positive integers.
    Then, there is a polynomial-time algorithm that finds $r$ forests $F_1, \ldots, F_r$ of $G$ with $|E(F_i)| = k$ for $1 \le i \le r$ such that $\sumham(E(F_1), \ldots, E(F_r))$ is maximized.
\end{corollary}

To give a precise running time bound, we use Cunningham's algorithm~\cite{Cunningham:Improved:1986} for the matroid intersection, and the total running time of our algorithm is $O((rR)^{1.5}r|E|Q)$, where $R$ is the rank of the matroid $\mathcal M$ and $Q$ is the running time of the independent oracle of $\mathcal M$.
In particular, for {\sc Diverse Spanning Tree}, it runs in time $O((rn)^{2.5}m)$, where $n$ and $m$ are the numbers of vertices and edges of the input graph, respectively.

\section{Finding Diverse $k$-Paths}\label{sec:k-path}
In the previous section, we present an efficient algorithm for finding diverse spanning trees.
A natural variant of this problem is to find diverse spanning paths.
However, the problem is clearly hard since the problem of finding a single spanning path, namely the Hamiltonian path problem, is NP-hard.
To cope with this difficulty, we investigate the complexity of this problem from the perspective of fixed-parameter tractability.
In this context, given a graph $G$ and integers $k$ and $r$, we are asked to find a set of $r$ paths of length $k - 1$ maximizing diversity measures.
Note the length of a path is defined to be the number of edges in the path.
In this section, we present an FPT algorithm for this problem.

\subsection{Finding a Single $k$-Path}
We first quickly review the algorithm of \cite{Alon:Color:1995} for finding a $k$-path (i.e., a path of length $k - 1$) in a graph $G = (V, E)$, which is widely known as the {\em color-coding} technique.
The algorithm basically consists of two steps: Assign one of $k$ colors to each vertex independently and uniformly at random and then seek a colorful $k$-path, that is, a path of length $k - 1$ whose vertices are colored with distinct colors.
By repeating these two steps sufficiently many times, we can decide whether $G$ has a $k$-path with high probability.

Suppose that $G$ has at least one $k$-path $P$.
We say that a coloring $c: V \to [k]$ is {\em good for $P$} if $P$ is $[k]$-colorful with respect to $c$.
Since the color of each vertex is determined independently and uniformly at random, the probability of a good coloring for $P$ is at least $k!/k^k = \Omega(e^{-k})$.
In vertex-colored graphs, we can find a colorful $k$-path by the following dynamic programming algorithm.
Let $c: V \to [k]$ be a coloring on vertices.
For each $v \in V$ and $C \subseteq [k]$, we say that a path $P$ is {\em feasible for $(C, v)$} if $P$ is $C$-colorful under $c$ and ends at $v$.
Given this, we wish to compute a feasible path for $([k], v)$ for some $v \in V$.
We can compute such a path by the following recurrence: there is a feasible path $P$ for $(C, v)$ of length at least one if and only if there is a feasible path for $(C \setminus \{c(v)\}, u)$ for some $u \in N_G(v)$.
This yields an algorithm for finding a colorful $k$-path that runs in time $O(2^{k}k(|V| + |E|))$, which is described in Algorithm~\ref{alg:k-path}.

\begin{algorithm}[t]                   
\caption{Given a vertex-colored graph $G = (V, E)$ with $c: V \to [k]$ and an integer $k$, find a colorful $k$-path in $G$}\label{alg:k-path}
\begin{algorithmic}[1]
    \State Set $\dpf(C, v) := $ {\bf False} for all $C \subseteq [k]$ and $v \in V$
    \For {$v \in V$}
        \State $\dpf(\{c(v)\}, v) := $ {\bf True}  
    \EndFor
    \For {$C \subseteq [k]$ in an increasing order of $|C|$}
        \For {$v \in V$ with $c(v) \in C$}
            \State $\dpf(C, v) := \displaystyle\bigvee_{u \in N_G(v)} \dpf(C \setminus \{c(v)\}, u)$
        \EndFor
    \EndFor
    \State Answer YES iff $\dpf([k], v) =$ {\bf True} for some $v \in V$
\end{algorithmic}
\end{algorithm}
By repeating the random coloring assignment and the algorithm for finding a colorful $k$-path (\Cref{alg:k-path}), we have the following randomized algorithm.
\begin{theorem}[\cite{Alon:Color:1995}]\label{thm:k-path}
    For every constant $\varepsilon > 0$, there is a Monte Carlo algorithm that decides whether $G$ has a $k$-path and runs in time $O((2e)^k(|V| + |E|))$, which does not contain false positive and returns a $k$-path with probability at least $\varepsilon$ if the answer is affirmative.
\end{theorem}

Alon et al. derandomize the algorithm of \Cref{thm:k-path} by substituting random coloring with {\em $k$-perfect hash family} \cite{Alon:Color:1995,Naor:Splitters:1995}.
Let $U$ be a finite set. A family $\mathcal H$ of functions $h: U \to [k]$ is called a {\em $k$-perfect hash family} if for every $k$-element subset $X \subseteq U$, there is a function $h \in \mathcal H$ such that $h(x) \neq h(y)$ for every $x, y \in X$ with $x \neq y$.
For any $U$ and a positive integer $k$, there is a $k$-perfect hash family of size $e^{k}k^{O(\log k)} \log |U|$ and it can be constructed in time $e^{k}k^{O(\log k)} |U| \log |U|$~\cite{Naor:Splitters:1995}.
Instead of using random coloring $c$, we use $k$-perfect hash family, and then we can deterministically decide whether $G$ has a $k$-path in time $(2e)^kk^{O(\log k)}|V|^{O(1)}$~\cite{Alon:Color:1995}.

\subsection{Extending to Diverse $k$-Paths}
Let $G = (V, E)$ be a graph and let $k$ and $r$ be positive integers.
We first present an algorithm for finding $r$ paths $P_1, \ldots, P_r$ of length $k - 1$ maximizing $\minham(V(P_1), \ldots, V(P_r))$ rather than $\minham(E(P_1), \ldots, E(P_r))$ (i.e., {\sc Min-Diverse $k$-Path}) for the reason that this version is conceptually simpler.
The algorithm for {\sc Min-Diverse $k$-Path} is postponed to the next section.
Although we only discuss an algorithm for maximizing $\minham(V(P_1), \ldots, V(P_r))$, the technique can be readily applied to the one for maximizing $\sumham(V(P_1), \ldots, V(P_r))$ as well.

To find a set of diverse $r$ paths of length $k - 1$, we leverage the same idea of \cite{Alon:Color:1995}.
We use $kr$ colors and randomly assign these colors to each vertex of $G$.
Since an optimal solution consists of at most $kr$ vertices (i.e., not necessarily vertex-disjoint $r$ paths of length $k - 1$), the probability of assigning distinct colors to each vertex in the solution is at least $(kr)!/(kr)^{kr} = \Omega(e^{-kr})$.
More specifically, let $\mathcal P = \{P_1, P_2, \ldots, P_r\}$ be a set of (not necessarily vertex-disjoint) $r$ paths of length $k - 1$.
We say that a coloring is {\em good for $\mathcal P$} if each vertex in $\bigcup_{P \in \mathcal P} V(P)$ receives a distinct color.
Suppose that $c : V \to [kr]$ is a good coloring for $\mathcal P$.
Then, each $k$-path in $\mathcal P$ is obviously colorful in this coloring.
Thus, we run Algorithm~\ref{alg:k-path} and compute the entry of $\dpf(C, v)$ for each $C \subseteq [kr]$ with $|C| \le k$ and $v \in V$, which is the existence of $C$-colorful paths in the colored graph.
Clearly, this can be done in time $\binom{kr}{k}n^{O(1)}$.

The vertex coloring not only gives a way to find a $k$-path in FPT time but also allows us to maximize the diversity of $k$-paths by enumerating sets of colors.
Fix a vertex-coloring $c: V \to [kr]$. 
We say that a set of colors $C \subseteq [kr]$ with $|C| = k$ is {\em feasible (under $c$)} if there is a $C$-colorful path in $G$.
The following observation is straightforward but is the heart of our algorithm.

\begin{observation}\label{obs:k-path:lb}
    Let $c: V \to [kr]$ be a coloring of $V$ and let $C_1, C_2, \ldots, C_r \subseteq [kr]$ be sets of colors with $|C_i| = k$ for $1 \le i \le r$.
    Suppose that $C_i$ is feasible for all $1 \le i \le r$.
    Then, there are (not necessarily vertex-disjoint and even not necessarily distinct) $r$ paths $P_1, \ldots, P_r$ of length $k - 1$ such that $\minham(V(P_1), \ldots, V(P_r)) \le \minham(C_1, \ldots, C_r)$.
\end{observation}

This implies that for every vertex-coloring $c: V \to [kr]$ and for feasible sets $C_1, \ldots, C_r \subseteq [kr]$ with respect to $c$, we have $OPT \ge \minham(C_1, \ldots, C_r)$, where $OPT$ is the optimal diversity of $r$ paths of length $k - 1$, and equality holds if $c$ is good for an optimal solution $\mathcal P = \{P_1, \ldots, P_r\}$ and $P_i$ is $C_i$-colorful for all $1 \le i \le r$.
With probability at least $e^{-kr}$, a random coloring is good for $\mathcal P$.
For any vertex-coloring $c$, we can check the feasibility of all color sets $C \subseteq [kr]$ with $|C| = k$ in total time $O(\binom{kr}{k}k(|V| +  |E|))$.
Therefore, we can compute $OPT$ in time $O(\binom{kr}{k}k(|V| + |E|) + \binom{kr}{k}^r(kr)^{O(1)})$ by simply enumerating all $r$-tuples of feasible color sets, assuming $c$ is good for $\mathcal P$.
Overall, the total running time is
    $O\left(e^{kr}\left(\binom{kr}{k}k(|V| + |E|) + \binom{kr}{k}^r(kr)^{O(1)}\right)\right)$,
which is $2^{O(kr \log (kr))}(|V| + |E|)$.

\begin{theorem}\label{thm:diverse-k-path}
    For every constant $\varepsilon > 0$, there is a Monte Carlo algorithm that finds $r$ paths of length $k - 1$ in $G$ maximizing $\minham(V(P_1), \ldots, V(P_r))$ or concludes that $G$ has no $k$-paths in time $2^{O(kr \log (kr))}(|V| + |E|)$.
    Moreover, it does not contain false positive and returns an optimal solution with probability at least $\varepsilon$ if $G$ has at least one $k$-path.
\end{theorem}

This can be derandomized by the $kr$-perfect hash family as well.

\begin{corollary}
    There is a deterministic algorithm that finds a set of $r$ paths of length $k - 1$ in $G$ maximizing $\minham(V(P_1), \ldots, V(P_r))$ or concludes that $G$ has no $k$-paths in time $2^{O(kr \log (kr))}|V|^{O(1)}$.
    Moreover, it does not contain false positive and returns an optimal solution with probability at least $\varepsilon$ if $G$ has at least one $k$-path.
\end{corollary}



\section{A General Framework for Finding Diverse Solutions}
In the previous section, we demonstrate a method for finding diverse $k$-paths using the color-coding technique due to Alon et al.~\cite{Alon:Color:1995}.
The power of this method is not limited to finding diverse $k$-paths.
Let $U$ be a finite set and let $\Pi: 2^U \to \{0, 1\}$ be an arbitrary function.
Throughout the paper, we assume that the function can be evaluated in time $|U|^{O(1)}$.
We call a subset $X$ of $U$  a {\em $\Pi$-set} if $\Pi(X) = 1$.
Many combinatorial objects can be expressed as this function (e.g., define $U$ as the set of vertices of a graph and $\Pi(X) = 1$ if and only if $X \subseteq U$ forms a path of length $k - 1$). 

The essence of the previous algorithm is that random coloring boils down the problem of finding diverse $\Pi$-sets to that of finding a single $C$-colorful $\Pi$-set for given color set $C$. 
More precisely, we have the following lemma.

\begin{lemma}\label{lem:gen:exact}
Suppose that, there is an algorithm $\mathcal A$ that, given a finite set $U$, a coloring $c: U \to [kr]$, and $C \subseteq [kr]$ with $|C| = k$, decides whether there is a $C$-colorful $\Pi$-set in time $f(k, r)|U|^{O(1)}$.
Then we can find $\Pi$-sets $U_1, \ldots, U_r$, with $|U_i| = k$ for $1 \le i \le r$ maximizing $\minham(U_1, \ldots, U_r)$ or conclude that there is no $\Pi$-set of size $k$ in $U$ in time $e^{rk}\left(\binom{kr}{k}f(k,r)|U|^{O(1)} + \binom{kr}{k}^r(kr)^{O(1)}\right)$ with high probability.
Moreover, a deterministic counterpart runs in time $2^{kr \log (kr)}f(k, r)|U|^{O(1)}$.
\end{lemma}

\begin{proof}
    The proof is almost analogous to that in the previous section and hence we only give a sketch of the proof.
    
    Let $U_1, \ldots, U_r$ be $\Pi$-sets of $U$ that maximize $\minham(U_1, \ldots, U_r)$.
    We assign one of the colors in $[kr]$ to each element in $U$ independently and uniformly at random.
    Then, with probability at least $e^{-kr}$, each element in $\bigcup_{1 \le i \le r} U_i$ receives a distinct color.
    For $C \subseteq [kr]$, we say that $C$ is {\em feasible (under $c$)} if there is a $C$-colorful $\Pi$-set of $U$.
    By using $\mathcal A$, we can check in time $\binom{kr}{k}f(k,r)|U|^{O(1)}$ the feasibility of all $C \subseteq [kr]$ with $|C| = k$.
    To compute $\minham(U_1, \ldots, U_r)$, we simply find $r$ feasible color sets $C_1, \ldots, C_r$ maximizing $\minham(C_1, \ldots, C_r)$ by an exhaustive search.
    The deterministic algorithm is obtained by using the $kr$-perfect hash family as well.
\end{proof}

Note again that the problem of maximizing $\sumham$ is also solvable in the claimed running time in \Cref{lem:gen:exact}.
In the rest of this section, we discuss some applications of \Cref{lem:gen:exact}.

\subsection{Diverse Matchings}
Let $G = (V, E)$ be a graph.
A {\em matching} is a set of edges $M \subseteq E$ such that there are no edges in $M$ sharing a common end vertex.
Since the problem of computing a matching of maximum size is solvable in polynomial time~\cite{Edmonds:Paths:1965}, one may expect that {\sc Diverse Matching} or {\sc Min-Diverse Matching} can be solved in polynomial time as well.
However, it is unlikely: Finding two edge-disjoint perfect matchings in cubic graphs is known to be NP-hard~\cite{Holyer:NP-completeness:1981}.
To overcome this difficulty, we consider the problem of finding diverse matchings of size $k$ and show that this is fixed-parameter tractable parameterized by $k + r$.

\begin{theorem}\label{thm:matching}
    {\sc Diverse Matching} and {\sc Min-Diverse Matching} can be solved in time $2^{O(kr \log (kr))}|V|^{O(1)}$. 
\end{theorem}
\begin{proof}
    Suppose that each edge of $G$ is colored with one of the colors $[kr]$.
    For each $C \subseteq [kr]$ with $|C| = k$, do the following.
    We first remove all the edges colored with a color not contained in $C$.
    Then, we apply the algorithm of finding a colorful matching of size $k$ due to~\cite{Gupta:parameterized:2019}, which runs in time $\left(\frac{1 + \sqrt{5}}{2}\right)^k|V|^{O(1)}$.
    Let $\Pi: 2^E \to \{0, 1\}$ be a function such that $\Pi(M) = 1$ if and only if $M$ is a matching of $G$.
    By \Cref{lem:gen:exact}, the statement follows.
\end{proof}

\subsection{Diverse Interval Scheduling}
We are given a set of tasks represented by intervals $\mathcal I = \{[a_i, b_i] : 1 \le i \le n\}$, where $[a, b]$ is the (closed) interval whose end points are $a \in \mathbb R$ and $b \in \mathbb R$ with $a \le b$.
A subset $\mathcal I' \subseteq \mathcal I$ is {\em feasible} if $\mathcal I'$ has no overlapping intervals, that is, $[a, b] \cap [a', b'] = \emptyset$ for distinct $[a, b], [a', b'] \in \mathcal I'$.
{\sc Interval Scheduling} is the problem of computing a maximum cardinality feasible $\mathcal I' \subseteq \mathcal I$.
This problem is also known as the maximum independent set problem on interval graphs and can be solved in polynomial time by a simple greedy algorithm.
We consider the diverse variants of {\sc Interval Scheduling}.
In {\sc Diverse Interval Scheduling} and {\sc Min-Diverse Interval Scheduling}, we are given a set $\mathcal I$ of intervals and integers $k$ and $r$.
The goals of the problems are to find $r$ sets of feasible intervals $\mathcal I_1, \ldots, \mathcal I_r \subseteq \mathcal I$ with $|\mathcal I_i| = k$ for $1 \le i \le r$ maximizing $\sumham(\mathcal I_1, \ldots, \mathcal I_r)$ and $\minham(\mathcal I_1, \ldots, \mathcal I_r)$, respectively.
To solve these problems in FPT time, by \Cref{lem:gen:exact}, it suffices to show that the ``colorful version'' can be solved in FPT time.
This problem is also known as {\sc Job Interval Selection} in the literature.
In addition to the input of {\sc Interval Scheduling}, we are also given a coloring function $c: \mathcal I \to [k]$.
The goal of {\sc Job Interval Selection} is to find a maximum cardinality set of colorful feasible intervals in $\mathcal I$.
Halld\'orsson and Karlsson~\cite{Halldorsson:Strip:2006} showed that this problem can be solved in time $2^k|\mathcal I|^{O(1)}$, yielding the following result together with \Cref{lem:gen:exact}.

\begin{theorem}\label{thm:interval}
    {\sc Diverse Interval Scheduling} and {\sc Min-Diverse Interval Scheduling} can be solved in time $2^{O(kr \log (kr))}|\mathcal I|^{O(1)}$. 
\end{theorem}

\subsection{Diverse $k$-Paths Revisited}
As we have discussed in the previous section, the problem of finding $r$ paths $P_1, \ldots, P_r$ of length $k - 1$ maximizing $\minham(V(P_1), \ldots, V(P_r))$ can be solved in FPT time.
To solve {\sc Diverse $k$-Paths} and {\sc Min-Diverse $k$-Paths}, we use Lemma~\ref{lem:gen:exact} and a modified version of Algorithm~\ref{alg:k-path}.
Suppose that we are given an edge-colored graph $G = (V, E)$ with $c_E : E \to [(k - 1)r]$.
For each $D \subseteq [(k - 1)r]$ with $|D| = k - 1$, we design an algorithm for finding a $k$-path whose edges are $D$-colorful.
The algorithm is quite similar to the $k$-path algorithm described in the previous section.
We assign a color from $[k]$ to each vertex of $G$ independently and uniformly at random.
Fix a vertex coloring $c_V : V \to [k]$.
For each $C \subseteq [k]$ and for each $D' \subseteq D$ with $|C| - 1 = |D'|$,
we say that a path $P$ is {\em $(C, D')$-colorful} if $V(P)$ is $C$-colorful and $E(P)$ is $D'$-colorful.
The following recurrence immediately yields a dynamic programming algorithm for $([k], D)$-colorful paths.
For each $C \subseteq [k]$ and $D' \subseteq D$ with $|C|-1 = |D| \ge 1$, and $v \in V$ with $c_V(v) \in C$, $G$ has a $(C, D')$-colorful path ending at $v$ if and only if there is a $(C \setminus \{c(v)\}, D' \setminus \{c_E(\{u, v\})\})$-colorful path ending at $u \in V$ with $\{u, v\} \in E$.
This recurrence can be straightforwardly implemented by dynamic programming whose running time is in $O(4^kk(|V| + |E|))$.
By using the $k$-perfect hash family for $c_v$, we finally conclude that {\sc Diverse $k$-Path} and {\sc Min-Diverse $k$-Path} are fixed-parameter tractable.

\begin{theorem}\label{thm:k-path:edge}
    {\sc Diverse $k$-Path} and {\sc Min-Diverse $k$-Path} can be solved in time $2^{O(kr \log (kr))}|V|^{O(1)}$, where $n$ is the number of vertices of the input graph.
\end{theorem}

\subsection{Diverse Subgraphs of Bounded Treewidth}
Let $G = (V, E)$ and $G' = (V', E')$ be graphs.
We say that $G$ is {\em isomorphic to $G'$} if there is a bijection $f$ between $V$ and $V'$ such that $\{u, v\} \in E$ if and only if $\{f(u), f(v)\} \in E'$ for every pair of vertices $u,v \in V$.
Given graphs $G$ and $H$, the subgraph isomorphism problem asks whether $G$ has a subgraph isomorphic to $H$.
This problem is a common generalization of many NP-hard problems, including the Hamiltonian path problem and the $k$-clique problem.
Since the $k$-clique problem, the problem of finding a clique of $k$ vertices, is W[1]-hard parameterized by $k$~\cite{Downey:Fixed:1995}, finding diverse subgraphs isomorphic to $H$ is also hard for general $H$.
When $H$ is ``sparse'', however, the problem is fixed-parameter tractable~\cite{Alon:Color:1995}.

A {\em tree decomposition} of $G = (V, E)$ is a pair $(T, \{X_i \subseteq V: i \in I\})$ of a rooted tree $T$ with node set $I$ and a collection $\{X_i : i \in I\}$ of subsets of $V$ such that
(1) $\bigcup_{i \in I} X_i = V$; (2) for each $\{u, v\} \in E$, there is an $i \in I$ with $\{u, v\} \subseteq X_i$; and (3) for each $v \in V$, the set of nodes $\{i \in I: v \in X_i\}$ induces a subtree of $T$.
The {\em width} of a tree decomposition $(T, \{X_i : i \in I\})$ is defined as $\max_{i \in I}|X_i| - 1$, and the {\em treewidth} of $G$ is the minimum integer $t$ such that $G$ has a tree decomposition of width $t$.
Alon et al.~\cite{Alon:Color:1995} showed that when the treewidth of $H$ is a constant, the subgraph isomorphism problem is fixed-parameter tractable.

\begin{theorem}[\cite{Alon:Color:1995}]\label{thm:iso-tw}
    Let $G$ and $H$ be graphs.
    Suppose that $|V(H)| = k$ and the treewidth of $H$ is $t$.
    Then, there is an algorithm that decides if $G$ has a subgraph isomorphic to $H$ in time $2^{O(k)}|V(G)|^{t+1}\log |V(G)|$.
\end{theorem}

We briefly describe their idea of \Cref{thm:iso-tw}.
Let $(T, \{X_i : i \in I)$ be a tree decomposition of $H$.
For each $i \in I$, we denote by $H_i$ the subgraph of $H$ induced by the vertices appeared in $X_i$ or $X_j$ for some descendant $j \in I$ of $i$.
They also use the color-coding technique as follows.
Let $c: V \to [k]$ be a coloring on $V$.
For each $C \subseteq [k]$, $i \in I$, $Z \subseteq V$ with $|Z| = |X_i|$, and a bijection $f_Z: Z \to X_i$, the algorithm decides whether $G$ has a $C$-colorful subgraph $H'$ isomorphic to $H_i$ such that $Z \subseteq V(H')$ and the bijection $f: V(H') \to V(H)$ extends $f_Z$, that is, $f(z) = f_Z(z)$ for all $z \in Z$.
This can be done in time $2^{O(k)}|V(G)|^{t+1}$ by dynamic programming and hence \Cref{thm:iso-tw} follows.

By combining this algorithm with \Cref{lem:gen:exact}, we have the following result.
\begin{theorem}
    Let $G$ and $H$ be graphs.
    Suppose that $|V(H)| = k$ and the treewidth of $H$ is $t$.
    Then there is a $2^{O(kr \log (kr))}|V(G)|^{t+O(1)}$-time algorithm that finds $r$ subgraphs $H_1, \ldots, H_r$ isomorphic to $H$ such that $\sumham(V(H_1), \ldots, V(H_r))$ or $\minham(V(H_1), \ldots, V(H_r))$ is maximized, or concludes $G$ has no subgraph isomorphic to $H$. 
\end{theorem}

We can extend this algorithm to {\sc Diverse Subgraph Isomorphism} and {\sc Min-Diverse Subgraph Isomorphism} by simultaneously considering edge-colorings as in \Cref{thm:k-path:edge}.
To be more precise, let $c_V: V \to [k]$ and $c_E: E \to [|E(H)|]$ be random vertex- and edge-colorings.
For each $C \subseteq [k]$, $D \subseteq [|E(H)|]$, $i \in I$, $Z \subseteq V$ with $|Z| = |X_i|$, and a function $f': Z \to X_i$, we can decide whether $G$ has a $(C, D)$-colorful subgraph $H'$ isomorphic to $H_i$ such that $Z \subseteq V(H')$ and the bijection $f: V(H') \to V(H)$ extends $f'$, that is, $f(z) = f'(z)$ for all $z \in Z$ by dynamic programming as well as the vertex-colored case.
Here, a subgraph $H'$ is $(C, D)$-colorful if $V(H')$ is $C$-colorful and $E(H')$ is $D$-colorful.
It is not hard to see that the running time of this dynamic programming algorithm is $2^{O(k + |E(H)|)}|V(G)|^{t+1}$.

\begin{theorem}
    {\sc Diverse Subgraph Isomorphism} and {\sc Min-Diverse Subgraph Isomorphism} can be solved in time $f(k,r)|V(G)|^{t+O(1)}$ for some function $f$, where $k = |V(H)|$, and $t$ is the treewidth of $H$.
\end{theorem}

\subsection{Diverse Subgraphs with FO Properties}
We have seen several examples for which finding diverse solutions is fixed-parameter tractable.
In particular, {\sc Diverse $k$-Path} and {\sc Min-Diverse $k$-Path} are fixed-parameter tractable parameterized by $k + r$.
In contrast to this tractability, finding a single {\em induced} $k$-path is known to be $W[2]$-complete~\cite{Chen:parameterized:2007} with respect to parameter $k$, where an {\em induced path} in a graph is a path such that non-consecutive vertices on the path are not adjacent in the graph.

A typical approach to overcoming such intractable problems is to restrict input graphs to some sparse graph classes, such as planar graphs, bounded-treewidth graphs, or $H$-minor free graphs with fixed $H$.
One of the most remarkable results in this context is an algorithmic metatheorem for sparse graphs with first-order logic (FO) given by Grohe et al.~\cite{Grohe:Deciding:2017}.
Many graph properties can be expressed by a formula in FO and, in particular, so is the property of being a colorful induced path of length $k - 1$.

In first-order logic formulas on graphs, we are allowed to use variables for vertices, basic logical connectives $\lor, \land, \neg, \implies, \iff$, and the universal and existential quantifiers $\forall, \exists$ for variables.
Moreover, unary and binary relational symbols on variables can be used, such as $x = y$, $E(x, y)$, and $C_i(x)$ that indicate equivalence of variables, adjacency of variables, and color indicator of variables, respectively.

For example, define a predicate $\phi_p(v_1, \ldots, v_k)$ as
\[
    \phi_p(v_1, \ldots, v_k) := \bigwedge_{1 \le i < k} E(v_i, v_{i+1}) \land \bigwedge_{\substack{1 \le i < j \le k\\
    j \neq i + 1}}\neg E(v_i, v_j).
\]
This indicates that vertices $v_1, \ldots v_k$ form an induced path.
We can also express ``colorful constraints'' using color indicators as
\[
    \phi_c(v_1, \ldots, v_k) :=\bigwedge_{\substack{1 \le i, j\le k\\ i \neq j}} \bigwedge_{1 \le \ell \le k} (C_{\ell}(v_i) \implies \neg C_{\ell}(v_j)),
\]
where $C_{\ell}(v_j)$ means ``vertex $v_j$ is colored with $\ell$''.
We let
\[
    \phi := \exists v_1, \ldots, v_k(\phi_p(v_1, \ldots, v_k) \land \phi_c(v_1, \ldots, v_k)).
\]
Then, it is easy to see that $G \models \phi$ if and only if $G$ has a colorful induced path of length $k - 1$ with respect to $c : V \to [k]$.
Grohe et al.~\cite{Grohe:Deciding:2017} showed that for a nowhere dense class $\mathcal C$ of graphs and an FO formula $\phi$, checking whether $G \models \phi$ is fixed-parameter tractable parameterized by the length $|\phi|$ of formula $\phi$.

\begin{theorem}[\cite{Grohe:Deciding:2017}]\label{thm:FO-fpt}
    Let $G = (V, E)$ be a graph that is in a nowhere dense class of graphs and $c: V \to [k]$ a coloring on vertices.
    Let $\Pi$ be a property on graphs that is expressible by a formula $\phi$ in first-order logic. 
    Then, for every $\varepsilon > 0$, one can decide whether $G$ has a $[k]$-colorful subgraph satisfying $\Pi$ in time $f(\varepsilon, |\phi|, k)|V|^{1+\varepsilon}$ for some function $f$.
\end{theorem}

We do not give the definition of nowhere dense classes of graphs (See~\cite{Grohe:Deciding:2017} for details), while it is worth mentioning that many sparse graph classes are included in these classes, such as planar graphs, bounded-treewidth graphs, and $H$-minor free graphs.
By combining \Cref{lem:gen:exact} and \Cref{thm:FO-fpt}, we have the following theorem.

\begin{theorem}\label{thm:diverse-FO}
    Let $G = (V, E)$ be a graph that is in a nowhere dense class of graphs.
    Let $\Pi$ be a property on graphs that is expressible by a formula $\phi$ in first-order logic. 
    Then, we can find $r$ subgraphs $H_1, \ldots, H_r$ of $k$ vertices satisfying $\Pi$ such that $\minham(V(H_1), \ldots, V(H_r))$ or $\sumham(V(H_1), \ldots, V(H_r))$ is maximized, in time $f(|\phi|, k, r)|V|^{O(1)}$.
\end{theorem}

\section{Acknowledgements}
The authors thank anonymous referees for valuable comments.
This work was partially supported by JSPS Kakenhi Grant Numbers JP18K11168, JP18K11169, JP18H04091, JP19J10761, JP19K21537, JP20K19742, and JP20H05793.

\bibliographystyle{plain}
\bibliography{ref}

\end{document}